\newcommand{\shortversion}[1]{}
\newcommand{\longversion}[1]{#1}

\shortversion{
\documentclass{llncs} 
}

\longversion{
\documentclass[10pt,usletter]{article}
\usepackage[totalwidth=450pt,totalheight=640pt]{geometry}
\date{}
\usepackage{amsthm}
}

\usepackage{times}

\shortversion{
\title{Upper and Lower Bounds for\\ Weak Backdoor Set
  Detection\thanks{All authors acknowledge support from the OeAD/DST
    (Austrian Indian collaboration grant, IN13/2011). Szeider
    acknowledges the support by the ERC, grant reference 239962.}}
\author{Neeldhara Misra\inst{1} \and Sebastian Ordyniak\inst{2} \and Venkatesh Raman\inst{3} \and Stefan Szeider\inst{4}}

\institute{
  Indian Institute of Science, Bangalore
  \and
  Masaryk University Brno
  \and
  Institute of Mathematical Sciences, Chennai
  \and
  Vienna University of Technology
}
}  

\longversion{
\title{Upper and Lower Bounds for Weak Backdoor Set Detection\thanks{All authors acknowledge support from the OeAD/DST
    (Austrian Indian collaboration grant, IN13/2011). Szeider
    acknowledges the support by the ERC, grant reference 239962.}}
\author{Neeldhara  Misra${}^1$ \quad Sebastian Ordyniak${}^2$ \quad
  Venkatesh Raman${}^3$ \quad Stefan Szeider${}^4$\\[8pt]
\small  ${}^1$ Department of Computer Science and Automation,
  Indian Institute of Science, Bangalore\\[-3pt]
\small  \texttt{mail@neeldhara.com}\\
\small ${}^2$  Faculty of Informatics,
  Masaryk University Brno,
 Czech Republic\\[-3pt]
\small  \texttt{ordyniak@fi.muni.cz}\\
\small ${}^3$  The Institute of Mathematical Sciences, Chennai\\[-3pt]
\small  \texttt{vraman@imsc.res.in}\\
\small ${}^4$ Institute of Information Systems,
  Vienna University of Technology, Austria\\[-3pt]
\small  \texttt{stefan@szeider.net}
}

}

\usepackage{url} 

\usepackage{amsmath}
\usepackage{amssymb}
\usepackage{booktabs}
\usepackage{tikz}
\usepackage{multirow}
\usepackage{enumitem}

\newcommand{\AAA}{\mathcal{A}}
\newcommand{\BBB}{\mathcal{B}}

\newcommand{\SSS}{\mathcal{S}}

\newcommand{\class}[1]{\text{\text{\normalfont\sc  #1}}}

\newcommand{\HORN}{\class{Horn}}
\newcommand{\ZEROV}{\class{$0$-Val}}

\newcommand{\RHORN}{\class{RHorn}}
\newcommand{\QHORN}{\class{QHorn}}
\newcommand{\MATCH}{\class{Match}}
\newcommand{\KROM}{\class{Krom}}
\newcommand{\CNF}{\class{CNF}}
\newcommand{\THREECNF}{\class{3CNF}}

\newcommand{\CHAINS}{\class{Chains}}
\newcommand{\ACYC}{\class{Forest}}

\newcommand{\hy}{\hbox{-}\nobreak\hskip0pt}
\newcommand{\brk}{\hspace{0pt}}

\newcommand{\SB}{\{\,}%
\newcommand{\SM}{\;{:}\;}%
\newcommand{\SE}{\,\}}%
\newcommand{\Card}[1]{|#1|}

\newcommand{\mtext}[1]{\text{\normalfont #1}}
\newcommand{\var}{\mtext{var}}

\newcommand{\FPT}{\text{\normalfont FPT}}

\newcommand{\W}[1][xxxx]{\text{\normalfont W}[#1]}

\newcommand{\ol}[1]{\overline{#1}}

\newtheorem{LEM}{Lemma} 
\newtheorem{THE}{Theorem} 
 
\newtheorem{COR}{Corollary} 
\newtheorem{CLM}{Claim}

\shortversion{
 
 \renewenvironment{proof}{\begin{pf}}{\qed\end{pf}}
 \newenvironment{proofnoqed}{\begin{pf}}{\end{pf}}
}

\let\phi=\varphi
\renewcommand{\phi}{\varphi}

\newcommand{\scite}[1]{${}^\text{{\protect\cite{#1}}}$}
\newcommand{\we}{${}^{[\star]}$}
\newcommand{\triv}{${}^{[\text{triv}]}$}
\begin{document}

\maketitle

\begin{abstract}\noindent
  We obtain upper and lower bounds for running times of exponential
  time algorithms for the detection of weak backdoor sets of 3CNF
  formulas, considering various base classes. These results include
  (omitting polynomial factors),
  (i)~a $4.54^k$ algorithm to detect whether there is a weak backdoor
  set of at most $k$ variables into the class of Horn formulas;
  (ii)~a $2.27^k$ algorithm to detect whether there is a weak
    backdoor set of at most $k$ variables into the class of Krom formulas.
These bounds improve an earlier known bound of $6^k$.  We
  also prove a $2^k$ lower bound for these problems, subject to the
  Strong Exponential Time Hypothesis.
\end{abstract}

\section{Introduction}
\longversion{\thispagestyle{empty}}
A backdoor set is a set of variables of a CNF formula such that fixing
the truth values of the variables in the backdoor set moves the
formula into some polynomial-time decidable class.  Backdoor sets
were independently introduced by Crama et
al.~\cite{CramaEkinHammer97} and by Williams et
al.~\cite{WilliamsGomesSelman03}, the latter authors coined the term
``backdoor.'' The existence of a small backdoor set in a CNF formula
can be considered as an indication of ``hidden structure'' in the
formula. 

One distinguishes between various types of backdoor sets. Let $\BBB$
denote the base class of formulas under consideration.  A \emph{weak}
\emph{$\BBB$-backdoor set} of a CNF formula $F$ is a set $S$ of
variables such that there is a truth assignment~$\tau$ of the variables in $S$ 
for which the formula $F[\tau]$, which is obtained from $F$ 
by assigning the variables of $S$ according to $\tau$ and applying the usual simplifications,
 is satisfiable and $F[\tau]\in \BBB$.  A \emph{strong} \emph{$\BBB$-backdoor
  set} of $F$ is a set~$S$ of variables such that for {\em each} truth
assignment $\tau$ of the variables in $S$, the formula $F[\tau]$ is in
$\BBB$.
 
The challenging problem is to find a weak or strong $\BBB$\hy backdoor
set of size at most~$k$ if it exists.  These problems are NP-hard for
all reasonable base classes. However, if~$k$ is assumed to be small,
an interesting complexity landscape evolves, which can  be adequately
analyzed in the context of \emph{parameterized complexity}, where~$k$
is considered as the parameter (some basic notions of parameterized
complexity will be reviewed in Section~\ref{section:prelim}).  This
line of research was initiated by Nishimura~et
al.~\cite{NishimuraRagdeSzeider04-informal} who showed that for the
fundamental base classes $\HORN$ and $\KROM$, the detection of strong
backdoor sets is fixed-parameter tractable, whereas the detection of
weak backdoor sets is not (under the complexity theoretic assumption
$\FPT\neq\W[1]$). However, if the width of the clauses of the input
formula is bounded by a constant, then these hardness results do not
hold any more and one achieves fixed-parameter
tractability~\cite{GaspersSzeider12}. In order to discuss these
results, we introduce the following problem template which is defined
for any two classes $\AAA,\BBB$ of CNF formulas.
\newcommand{\WB}{\text{\normalfont WB}}
\begin{quote}
  $\WB(\AAA,\BBB)$\nopagebreak
  
  \emph{Instance:} A CNF formula $F\in \AAA$ with $n$ variables, a
  non-negative integer~$k$.\nopagebreak

  \emph{Parameter:} The integer $k$.\nopagebreak

  \emph{Question:} Does $F$ have a weak $\BBB$\hy backdoor set of size at
  most $k$?
\end{quote}
Thus, one could think of this problem as asking for a small weak
backdoor ``from $\AAA$ to $\BBB$.'' In this paper we focus on the
special case where $\AAA=\THREECNF$. In particular, we aim to draw
a detailed complexity landscape of $\WB(\THREECNF,\BBB)$ for various
base classes, providing improved lower and upper bounds. An overview
of our results in the context of known results is provided in 
Table~\ref{tab:results}. The
definitions of these classes appear in Section~\ref{section:prelim}.
\begin{table}[tb]
  \centering 
\begin{tabular}{@{}l@{~~~~}c@{~~~}c@{~~~}c@{~~~}c@{~~~}c@{~~~~~}c@{~~~}c@{}}
\toprule
$\BBB$:   & $\HORN$ & $\KROM$ &  $\ZEROV$ & $\ACYC$ & $\RHORN$ & $\QHORN$ & $\MATCH$ \\
\midrule 
UB: &  
$4.54^k$ \we &
$2.27^k$ \we &
$2.85^k$ \scite{RamanShankar13} &
$f(k)$ \scite{GaspersSzeider12b} &
$n^k$ \triv & 
$n^k$ \triv & 
$n^k$ \triv \\  
LB: &
$2^k$ \we &
$2^k$ \we &
$2^{o(k)}$ \we &
$2^k$ \we &
$n^{\frac{k}{2}-\epsilon}$ \scite{GaspersSzeider12} &
$n^{\frac{k}{2}-\epsilon}$ \scite{GaspersEtal13} &
$n^{\frac{k}{2}-\epsilon}$ \we \\
\bottomrule \\
\end{tabular}
\caption{Upper bounds (UB) and lower bounds (LB) for the time complexity of $\WB(\THREECNF,\BBB)$ for various base
  classes $\BBB$ (polynomial
  factors are omitted). The $2^k$ and $n^{\frac{k}{2}-\epsilon}$ 
  lower bounds are subject to the Strong
  Exponential-Time Hypothesis, and the $2^{o(k)}$ lower bounds are subject to the 
  Exponential-Time Hypothesis.
  Results marked~\we\ are obtained in this paper.
} 
  \label{tab:results}
\end{table}

\sloppypar Gaspers and Szeider \cite{GaspersSzeider12} showed that
$\WB(\THREECNF,\BBB)$ is fixed-parameter tractable for every base
class $\BBB$ which is defined by a property of individual clauses,
such as the classes $\HORN$, $\KROM$, and $\ZEROV$. Their general
algorithm provides a running time of~$6^k$ (omitting polynomial
factors).  We improve this to $4.54^k$ for $\HORN$ and to $2.27^k$ for
$\KROM$. These results fit nicely with the recent $2.85^k$ algorithm
for $\WB(\THREECNF,\ZEROV)$ by Raman and
Shankar~\cite{RamanShankar13}.

There are base classes for which the detection of weak backdoor sets
remains fixed-parameter \emph{in}tractable (in
terms of $\W[2]$\hy hardness), even if the input is
restricted to 3CNF.  In particular, the $\W[2]$-hardness of $\WB(\THREECNF,\BBB)$ is known for
the base class \RHORN{}~\cite{GaspersSzeider12}
and for the base class of \QHORN{}~\cite{GaspersEtal13}:
We extend this line of results with another example. We consider the class
$\MATCH$ of \emph{matched} CNF formulas~\cite{FrancoVanGelder03},
which are CNF formulas~$F$ where for each clause $C\in F$ one can
select a unique variable $x_C$ that appears in $C$ positively or
negatively, such that $x_C\neq x_D$ for $C\neq D$. Since all matched
formulas are satisfiable, this class is particularly well suited as
a base class for weak backdoor sets. It is known that
$\WB(\CNF,\MATCH)$ is $\W[2]$\hy hard, but the case
$\WB(\THREECNF,\MATCH)$ has been open.  We show, that
$\WB(\THREECNF,\MATCH)$ is $\W[2]$\hy hard as well.

We contrast the algorithmic upper bounds for the considered backdoor set
detection problems by lower bounds. These lower bounds are either
subject to the Exponential Time Hypothesis (ETH), or the Strong
Exponential Time Hypothesis (SETH), see Section~\ref{sec:lb}.
Consequently, any algorithm that beats these lower bounds would provide an
unexpected speedup for the exact solution of 3SAT or SAT,
respectively.  In particular, we explain how the $\W[2]$\hy hardness
proofs can be used to get lower bounds of the form
$n^{\frac{k}{2}-\epsilon}$ under the SETH.

\shortversion{\paragraph{Full Version} Proofs of statements
  marked with ($\star$) are shortened or omitted due to space restrictions.
Detailed proofs can be found in the full
  version, available at \url{arxiv.org/abs/1304.5518}.}

\section{Preliminaries}\label{section:prelim}

\paragraph{CNF Formulas and Assignments}
We consider propositional formulas in conjunctive normal form (CNF) as sets
of clauses, where each clause is a set of literals, i.e., a literal is either 
a (positive) variable or a negated variable, not containing a
pair of complementary literals. We say that a variable $x$ is positive
(negative) in a clause $C$ if $x\in C$ ($\ol{x}\in C$), and we write
$\var(C)$ for the set of variables that are positive or negative in $C$.  A
\emph{truth assignment} $\tau$ is a mapping from a set of
variables, denoted by $\var(\tau)$, to $\{0,1\}$.  A truth assignment $\tau$ \emph{satisfies} a
clause $C$ if it sets at least one positive variable of $C$ to $1$
or at least one negative variable of $C$ to $0$. A truth assignment
$\tau$ satisfies a CNF formula~$F$ if it satisfies all clauses
of~$F$. Given a CNF formula $F$ and a truth assignment~$\tau$,
$F[\tau]$ denotes the \emph{truth assignment reduct} of $F$ under
$\tau$, which is the CNF formula obtained from~$F$ by first removing
all clauses that are satisfied by $\tau$ and then removing from the
remaining clauses all literals $x,\ol{x}$ with $x\in \var(\tau)$. Note that no
assignment satisfies the empty clause.
The \emph{incidence graph} of a CNF formula
$F$ is the
bipartite graph whose vertices are the variables and clauses of $F$,
and where a variable $x$ and a clause $C$ are adjacent if and only
if $x\in \var(C)$.

\medskip \noindent We consider the following \emph{classes of CNF formulas}.
\begin{itemize}
\item $\THREECNF$: the class of CNF formulas where each clause contains at
  most $3$ literals. 
\item $\KROM$: the class of CNF formulas where each clause contains at
  most $2$ literals (also called 2CNF). 
\item $\HORN$: the class of \emph{Horn} formulas, i.e., CNF formulas
  where each clause has at most $1$ positive literal.
\item $\RHORN$: the class of \emph{renameable} (or \emph{disguised})
  \emph{Horn} formulas, i.e., formulas that can be made Horn by
complementing variables.
\item $\QHORN$: the class of \emph{q-Horn}
    formulas~\cite{BorosHammerSun94} ($\RHORN,\KROM \subseteq
  \QHORN$).
\item $\ZEROV$: the class of \emph{$0$\hy valid}  CNF formulas, i.e.,
  formulas where
  each clause contains at least $1$ negative literal.
\item $\ACYC$: the class of \emph{acyclic} formulas (the undirected incidence graph
 is acyclic).
\item $\MATCH$: the class of \emph{matched} formulas, formulas whose
  incidence graph has a matching such that each clause is matched
  to some unique variable.
\end{itemize}
All our results concerning the classes $\HORN$ and $\ZEROV$ clearly
hold also for  the dual classes of \emph{anti-Horn} formulas (i.e., CNF formulas
where each clause has at most $1$ positive literal), and \emph{$1$\hy
  valid} CNF formulas (i.e., formulas where each clause contains at
least $1$ positive literal), respectively.

\paragraph{Parameterized Complexity}

Here we introduce the relevant concepts of parameterized complexity theory.
For more details, we refer to text books on the topic~\cite{DowneyFellows99,FlumGrohe06,Niedermeier06}.
An instance of a parameterized problem is a pair $(I,k)$ where $I$ is
the main part of the instance, and $k$ is the parameter.  A
parameterized problem is \emph{fixed-parameter tractable} if instances
$(I,k)$ can be solved in time $f(k)\Card{I}^c$, where $f$ is a
computable function of $k$, and $c$ is a constant.  $\FPT$ denotes the
class of all fixed-parameter tractable problems.  Hardness for
parameterized complexity classes is based on \emph{fpt-reductions}.  A
parameterized problem $L$ is fpt-reducible to another parameterized
problem~$L'$ if there is a mapping $R$ from instances of $L$ to
instances of $L'$ such that (i) $(I,k) \in L$ if and only if $(I',k')
= R(I,k) \in L'$, (ii) $k' \leq g(k)$ for a computable function $g$,
and (iii) $R$ can be computed in time $O(f(k)\Card{I}^c)$ for a
computable function $f$ and a constant~$c$.  
Central to the
completeness theory of parameterized complexity is the hierarchy $\FPT
\subseteq \W[1] \subseteq \W[2] \subseteq \dots $.  Each
intractability class $\W[t]$ contains all parameterized problems that
can be reduced to a certain parameterized satisfiability problem under
fpt-reductions.

\longversion{
The following problem is well-known to be $\W[2]$\hy complete.
\begin{quote}
  \textsc{Hitting Set}
  
  \emph{Instance:} A family $\SSS$ of finite sets $S_1,\dotso,S_m$ and an integer $k > 0$.
  
  \emph{Parameter:} The integer $k$.  
  
  \emph{Question:} Does $\SSS$ have a hitting set of size at most
  $k$, i.e., a set $H \subseteq \bigcup_{1 \leq i \leq m}S_i$ such
  that $H \cap S_i\neq \emptyset$ for every $1 \leq i \leq m$ and
  $|H|\leq k$?
\end{quote}} \shortversion{ The problem \textsc{Hitting Set} is
well-known to be $\W[2]$\hy complete.  This problem takes as input a
family $\SSS$ of finite sets $S_1,\dotso,S_m$ and an integer $k > 0$,
$k$ is the parameter.  The question is whether $\SSS$ has a hitting
set of size at most $k$, i.e., a set $H \subseteq \bigcup_{1 \leq i
  \leq m}S_i$ such that $H \cap S_i\neq \emptyset$ for every $1 \leq i
\leq m$ and $|H|\leq k$?
}

However, the restricted variant where all sets $S_i$ are of size at
most $3$, is fixed-parameter tractable and can be solved in time
$2.270^k$, omitting polynomial
factors~\cite{NiedermeierRossmanith03}.

\newpage 
  
\section{Upper Bounds}

\begin{THE}
   $\WB(\THREECNF,\KROM)$ can be solved in time $2.270^k$ (omitting
   polynomial factors).
\end{THE}
\begin{proof}
  Let $F$ and $k$ be the given $\THREECNF$ formula and non-negative
  integer, respectively.  Let $\SSS$ be the family of sets $\SB
  \var(C) \SM C \in F, |C|=3 \SE$. 
  We can
  find a weak $\KROM$\hy backdoor set of size at most~$k$ by finding a
  hitting set $H$ of $\SSS$ of size at most~$k$ and checking whether
  there is an assignment $\tau_H$ to the variables in $H$ such that
  $F[\tau_H]$ is satisfiable.  The correctness follows from the fact
  that if $F[\tau_H]$ is satisfiable for some $\tau_H$, then we clearly
  have the desired backdoor set. On the other hand, if $F[\tau_H]$ is
  not satisfiable for {\em any} $\tau_H$, then $F$ was not satisfiable
  to begin with, and does not admit a weak backdoor set of any size.
  As $F[\tau_H] \in \KROM$, the satisfiability of $F[\tau_H]$ can be
  checked in polynomial-time. It follows that if we omit polynomial
  factors then the running time of this algorithm is the time required
  to find a hitting set of $\SSS$ of size at most $k$, i.e.,
  $2.270^k$~\cite{NiedermeierRossmanith03}, plus the time required to
  go over the at most~$2^k$ assignments of the variables in the
  hitting set.
\end{proof}

\shortversion{
\begin{THE}[$\star$]\label{the:UB-HORN}
   $\WB(\THREECNF,\HORN)$ can be solved in time 
   $(\frac{1}{2}(1+\sqrt{65}))^k < 4.54^k$ (omitting
   polynomial factors).
\end{THE}
\begin{proofnoqed}
}
\longversion{
\begin{THE}\label{the:UB-HORN}
   $\WB(\THREECNF,\HORN)$ can be solved in time 
   $(\frac{1}{2}(1+\sqrt{65}))^k < 4.54^k$ (omitting
   polynomial factors).
\end{THE}
\begin{proof}
}
  Let $F$ and $k$ be the given $\THREECNF$ formula and non-negative
  integer, respectively. 
  If $F \in \HORN$ then there is nothing
  to do. So suppose that $F \notin \HORN$ and let $\textup{NH}(F)$ be the set
  of all clauses of $F$ that are not horn. Then $\textup{NH}(F)$ can
  contain the
  following types of clauses: (C1) clauses that contain only positive
  literals (and at least two of them), and (C2)
  clauses that contain exactly two positive literals and one negative
  literal. 

  An assignment $\tau$ is {\em minimal} with respect to a formula $F$ or
  to a
  clause $C$ if $F[\tau] \in \HORN$ or $\{C[\tau]\} \in
  \HORN$, respectively, but
  $F[\tau'] \notin \HORN{}$ or $\{C[\tau']\} \notin \HORN$ for every
  assignment $\tau'$ that agrees with
  $\tau$ but is defined on a strict subset of $\var(\tau)$.
  Our algorithm uses the bounded search tree method to 
  branch over all possible minimal assignments $\tau$ that set at
  most $k$
  variables of $F$ such that $F[\tau]\in \HORN$. 
  The algorithm
  then checks for each of these assignments whether $F[\tau]$ is
  satisfiable (because $F[\tau]\in \HORN$ this can be done in
  polynomial-time). If there is at least one such assignment $\tau$ such that
  $F[\tau]$ is satisfiable, then the algorithm returns $\var(\tau)$ as a weak $\HORN$\hy
  backdoor set of $F$ with witness $\tau$. Otherwise, i.e., if there is no
  such assignment, the algorithm outputs that~$F$
  does not have a weak $\HORN$\hy backdoor set of size at most $k$.

  At the root node of the search tree we set $\tau$ to be the empty
  assignment. Depending on the types and structure of the clauses of
  the formula $F[\tau]$, the algorithm then
  branches as follows:
  If $F[\tau]$ contains at least one clause of type (C1), then the
  algorithm branches on one of these clauses according to branching
  rule (R1). If $F[\tau]$ contains at least two clauses of
  type (C2) that are not variable-disjoint, then the algorithm
  branches on such a pair according to branching rule (R2). Otherwise,
  i.e., if $\textup{NH}(F[\tau])$ merely consists of clauses of type (C2) which
  are pairwise variable-disjoint the
  algorithm branches according to branching rule (R3).

  We will now describe the branching rules (R1)--(R3) in detail.
  In the following let~$\tau'$ be the assignment obtained before the
  current node in the search tree, and let $x$, $x'$, $y$,
  $y'$, $z$ and $z'$ be $6$ pairwise distinct variables. Every
  branching rule will lead to a new assignment $\tau$ (extending the current
  assignment $\tau'$) where the parameter~$k$ decreases by $|\var(\tau)
  \setminus \var(\tau')|$.
  
  Let $C \in F[\tau']$ be a clause of type (C1). Then branching rule
  (R1) is defined as follows.
  If $C=\{x,y,z\}$
  then $\{C[\tau]\} \in \HORN$ if and only if $\tau(x)=1$ or
  $\tau(y)=1$ or $\tau(z)=1$ or $\tau(x)=0 =\tau(y)$ or $\tau(x)=0=\tau(z)$, or $\tau(y)=0=\tau(z)$. Hence, there are
  $3$ cases for which the parameter (the number of variables set in the backdoor) decreases by $1$ and $3$ cases for which
  the parameter decreases
  by $2$. This leads to the recurrence relation 
  $T(k)=3T(k-1)+3T(k-2)=(\frac{1}{2}(3+\sqrt{21}))^k < 4.54^k$. 
  Similarly, if $C=\{x,y\}$
  then $\{C[\tau]\} \in \HORN$ if and only if $\tau(x)=0$ or
  $\tau(x)=1$ or $\tau(y)=0$ or $\tau(y)=1$. Hence, there are
  $4$ cases and in each of them the parameter decreases by $1$.
  This leads to the recurrence function $T(k)=4T(k-1)=4^k < 4.54^k$.

  Let $C \in F[\tau']$ and $C' \in F[\tau']$ be two distinct clauses of
  type (C2) that share at least one variable. Then branching rule
  (R2) is defined as follows.

  We distinguish the following cases\longversion{.} 
  \shortversion{(due to space
    limitations we will only list the cases here; for a full 
    description of the cases we refer to the full version of the paper).}
  \shortversion{(1A) $\var(C) \cap \var(C')=\{x\}$, $x \in C$ and $x
    \in C'$, (1B) $\var(C) \cap \var(C')=\{x\}$, $\bar{x} \in C$ and $\bar{x}
    \in C'$, (1C) $\var(C) \cap \var(C')=\{x\}$, $\bar{x} \in C$ and
    $x \in C'$, (2A) $\var(C) \cap \var(C')=\{x,y\}$, $x,y \in C$ and
    $x,y \in C'$, (2B) $\var(C) \cap \var(C')=\{x,y\}$, $x,\bar{y} \in C$ and
    $x,\bar{y} \in C'$, (2C) $\var(C) \cap \var(C')=\{x,y\}$,
    $x,\bar{y} \in C$ and $\bar{x},y \in C'$, and (3) $\var(C) \cap
    \var(C')=\{x,y,z\}$.
  }
  \longversion{
    \begin{itemize}
    \item \emph{($C$ and $C'$ have exactly $1$ variable $x$ in common)}

      We distinguish $3$ cases:
      \begin{itemize}
      \item\emph{($x$ is positive in $C$ and $C'$)}
        
        \longversion{
          In this case $C=\{x,y,\bar{z}\}$ and $C'=\{x,y',\bar{z'}\}$ and
          $\{C[\tau]\}, \{C'[\tau]\} \in \HORN$ if and only if one
          of the following holds:
          \begin{itemize}
          \item $\tau(x) \in \{0,1\}$;
          \item $\tau(y) \in \{0,1\}$ and $\tau(y') \in \{0,1\}$;
          \item $\tau(y) \in \{0,1\}$ and $\tau(z')=0$;
          \item $\tau(z)=0$ and $\tau(y') \in \{0,1\}$;
          \item $\tau(z)=0$ and $\tau(z')=0$.
          \end{itemize}
          This leads to the recurrence function
          $T(k)=2T(k-1)+9T(k-2)=(1+\sqrt{10})^k < 4.54^k$.
        }
      \item\emph{($x$ is negative in $C$ and $C'$)}
        
        \longversion{
          In this case $C=\{\bar{x},y,z\}$ and $C'=\{\bar{x},y',z'\}$ and
          $\{C[\tau]\}, \{C'[\tau]\} \in \HORN$ if and only if one
          of the following holds:
          \begin{itemize}
          \item $\tau(x)=0$;
          \item $\tau(y) \in \{0,1\}$ and $\tau(y') \in \{0,1\}$;
          \item $\tau(y) \in \{0,1\}$ and $\tau(z') \in \{0,1\}$;
          \item $\tau(z)=\{0,1\}$ and $\tau(y') \in \{0,1\}$;
          \item $\tau(z)=\{0,1\}$ and $\tau(z')=\{0,1\}$.
          \end{itemize}
          This leads to the following recurrence function: 
          $T(k)=T(k-1)+16T(k-2)=(\frac{1}{2}(1+\sqrt{65}))^k < 4.54^k$.
        }
      \item\emph{($x$ is negative in $C$ and positive in $C'$)}

        \longversion{
          In this case $C=\{\bar{x},y,z\}$ and $C'=\{x,y',\bar{z'}\}$ and
          $\{C[\tau]\}, \{C'[\tau]\} \in \HORN$ if and only if one
          of the following holds:
          \begin{itemize}
          \item $\tau(x)=0$;
          \item $\tau(x)=1$ and $\tau(y) \in \{0,1\}$;
          \item $\tau(x)=1$ and $\tau(z) \in \{0,1\}$;
          \item $\tau(y) \in \{0,1\}$ and $\tau(y') \in \{0,1\}$;
          \item $\tau(y) \in \{0,1\}$ and $\tau(z') =0$;
          \item $\tau(z) \in \{0,1\}$ and $\tau(y') \in \{0,1\}$;
          \item $\tau(z) \in \{0,1\}$ and $\tau(z')=0$.
          \end{itemize}
          This leads to the recurrence function
          $T(k)=T(k-1)+16T(k-2)=(\frac{1}{2}(1+\sqrt{65}))^k < 4.54^k$.
        }
      \end{itemize}
    \item \emph{($C$ and $C'$ have exactly $2$ variables $x$ and $y$ in common)}

      We distinguish $3$ cases:
      \begin{itemize}
      \item\emph{($x$ and $y$ are positive in $C$ and $C'$)}
        
        \longversion{
          In this case $C=\{x,y,\bar{z}\}$ and $C'=\{x,y,\bar{z'}\}$ and
          $\{C[\tau]\}, \{C'[\tau]\} \in \HORN$ if and only if one
          of the following holds:
          \begin{itemize}
          \item $\tau(x) \in \{0,1\}$;
          \item $\tau(y) \in \{0,1\}$;
          \item $\tau(z) = 0$ and $\tau(z') = 0$;
          \end{itemize}
          This leads to the recurrence function
          $T(k)=4T(k-1)+T(k-2)=(2+\sqrt{5})^k < 4.54^k$.
        }
      \item\emph{($x$ is positive in $C$ and $C'$ and $y$ is negative in $C$ and $C'$)}
        
        \longversion{
          In this case $C=\{x,\bar{y},z\}$ and $C'=\{x,\bar{y},z'\}$ and
          $\{C[\tau]\}, \{C'[\tau]\} \in \HORN$ if and only if one
          of the following holds:
          \begin{itemize}
          \item $\tau(x) \in \{0,1\}$;
          \item $\tau(y)=0$;
          \item $\tau(z) \in \{0,1\}$ and $\tau(z') \in \{0,1\}$.
          \end{itemize}
          This leads to the recurrence function
          $T(k)=3T(k-1)+4T(k-2)=4^k < 4.54^k$.
        }
      \item\emph{($x$ is positive in $C$ and negative in $C'$ and $y$ is
          negative in $C$ and positive in $C'$)}

        \longversion{
          In this case $C=\{x,\bar{y},z\}$ and $C'=\{\bar{x},y,z'\}$ and
          $\{C[\tau]\}, \{C'[\tau]\} \in \HORN$ if and only if one
          of the following holds:
          \begin{itemize}
          \item $\tau(x)=0$;
          \item $\tau(y)=0$;
          \item $\tau(x)=1$ and $\tau(y)=1$;
          \item $\tau(x)=1$ and $\tau(z') \in \{0,1\}$;
          \item $\tau(y)=1$ and $\tau(z) \in \{0,1\}$;
          \item $\tau(z)\in \{0,1\}$ and $\tau(z') \in \{0,1\}$.
          \end{itemize}
          This leads to the recurrence function: 
          $T(k)=2T(k-1)+9T(k-2)=(1+\sqrt{10})^k < 4.54^k$.
        }
      \end{itemize}
    \item \emph{($C$ and $C'$ have exactly $3$ variables $x$, $y$ and $z$ in common)}
      
      \longversion{
        W.l.o.g. let $x$ be the variable that occurs positively in $C$
        and $C'$. Then $C=\{x,y,\bar{z}\}$ and $C'=\{x,\bar{y},z\}$ and
        $\{C[\tau]\}, \{C'[\tau]\} \in \HORN$ if and only if one of
        the following holds:
        \begin{itemize}
        \item $\tau(x) \in \{0,1\}$;
        \item $\tau(y)=0$;
        \item $\tau(z)=0$;
        \item $\tau(y)=1$ and $\tau(z)=1$.
        \end{itemize}
        This leads to the recurrence function
        $T(k)=4T(k-1)+T(k-2)=(2+\sqrt{5})^k < 4.54^k$.
      }
    \end{itemize}
  }

 Taking the maximum over the above cases we obtain the recurrence
 function $T(k)=T(k-1)+16T(k-2)=(\frac{1}{2}(1+\sqrt{65}))^k < 4.54^k$
 for branching rule~(R2).

 Recall that after applying branching rule (R2) exhaustively all pairs
 of clauses of type (C2) are pairwise variable-disjoint. We will
 describe branching rule (R3), which makes use of this fact.
 Let $C=\{x,y,\bar{z}\}$. 
 We set
 either $\tau(x)=0$ or $\tau(x)=1$.
 This leads to the recurrence function: 
 $T(k)=2T(k-1)=2^k < 4.54^k$. Note that in contrast to the branching rules
 (R1) and (R2) the branching rule (R3) is not exhaustive. 
 Indeed for every clause $C=\{x,y,\bar{z}\}$ of type
 (C2) there are $5$
 possible minimal assignments~$\tau$ such that $\{C[\tau]\} \in \HORN$,
 i.e., the assignments $\tau(x)=0$, $\tau(x)=1$, $\tau(y)=0$,
 $\tau(y)=1$, and $\tau(z)=0$. Because each of these assignments
 $\tau$ sets only $1$ variable this would lead to a recurrence function
 $T(k)=5T(k-1)=5^k$ and hence $T(k)>4.54^k$. It follows that in
 contrast to the branching rules (R1) and (R2) where we could exhaustively
 branch over all possible minimal assignments, this cannot be done for
 clauses of type~(C2). However, because branching rule (R2) ensures
 that the remaining clauses of type~(C2) are pairwise
 variable-disjoint it turns out that this is indeed not necessary (see
 Claim~\ref{clm:UB-HORN}). 

 This concludes the description of our algorithm. The running time of
 the algorithm is the maximum branching factor over the cases
 described above, i.e., $(\frac{1}{2}(1+\sqrt{65}))^k < 4.54^k$ as
 required. To see that the algorithm is correct we need to show that
 it outputs an assignment $\tau$ if and only if the set $\var(\tau)$
 is a weak $\HORN$\hy backdoor set of $F$ of size at most $k$. Because
 the branching rules (R1) and (R2) branch exhaustively over all minimal
 assignments $\tau$ such that the corresponding clause(s) are
 reduced to Horn clauses, it
 only remains to show the correctness of branching rule (R3). 
\longversion{This is
 done by the following claim.
 \begin{CLM}\label{clm:UB-HORN}
   Let $F$ be a $\THREECNF$ formula, $P$ be a set of pairwise
   variable-disjoint clauses of type (C2) such that $F \setminus P
   \in \HORN$. Furthermore, let $L$ be a set of variables
   that consists of one positively occurring variable from each of
   the clauses in $P$. Then $F$ has a weak $\HORN$\hy backdoor set
   of size at most $|P|$ if and only if $L$ is a weak $\HORN$\hy
   backdoor set of $F$.
 \end{CLM}}
\shortversion{This is
 done by the following claim whose proof can be found in the full
 version of the paper.
 \begin{CLM}[$\star$]\label{clm:UB-HORN}
   Let $F$ be a $\THREECNF$ formula, $P$ be a set of pairwise
   variable-disjoint clauses of type (C2) such that $F \setminus P
   \in \HORN$. Furthermore, let $L$ be a set of variables
   that consists of one positively occurring variable from each of
   the clauses in $P$. Then $F$ has a weak $\HORN$\hy backdoor set
   of size at most $|P|$ if and only if $L$ is a weak $\HORN$\hy
   backdoor set of $F$.\qed
 \end{CLM}}
\longversion{To prove the claim, suppose that $F$ has a weak $\HORN$\hy
 backdoor set $B$ of size at most~$|P|$.
 Hence, there is an assignment $\tau^s$ that satisfies $F$. Let
 $\tau_L$ be an assignment of the variables of $L$ that agrees with
 $\tau^s$. Then $F[\tau_L] \in \HORN$ (this actually holds for every
 assignment of the variables in $L$) and $\tau_L$ can be extended to a
 satisfying assignment of~$F$. Hence, $L$ is a weak $\HORN$\hy
 backdoor set of $F$ of size at most $|L|=|P|$ with witness~$\tau_L$ as
 required. The reverse direction follows from the fact that $|L|=|P|$.    
Hence the claim is established and the theorem follows.}
\longversion{\end{proof}}
\shortversion{\end{proofnoqed}}

\section{Lower Bounds}\label{sec:lb}

For our lower bounds we use the \emph{Exponential Time Hypothesis}
(ETH) and the \emph{Strong Exponential Time Hypothesis}
(SETH), introduced by Impagliazzo et
al.~\cite{ImpagliazzoPaturi01,ImpagliazzoPaturiZane01}, which state
the following:
\begin{quote}\textbf{ETH}:
  There is no algorithm that decides the satisfiability of a
  $\THREECNF$ formula with $n$ variables in time $2^{o(n)}$, omitting
  polynomial factors.

\textbf{SETH}:
  There is no algorithm that decides the satisfiability of a
  $\CNF$ formula with $n$ variables in time 
  $(2-\epsilon)^n$, omitting polynomial factors.
\end{quote}
An \emph{implication chain} is a CNF formula of the form $\{
\{x_0\}$, $\{\bar{x}_0,x_1\}$, $\{\bar{x}_1,x_2\}$,\brk $\dots$,\brk $\{\bar{x}_{n-1},x_n\}$,
$\{\bar{x}_n \}\}$, $n\geq 1$ where the first $\{x_0\}$ and the last
clause $\{\bar{x}_n\}$ can be missing. Let $\CHAINS$ denote the class of
formulas that are variable-disjoint unions of implication chains.

  \begin{THE}
    Let $\BBB$ be a base class that contains $\CHAINS$. Then
    $\WB(\THREECNF,\BBB)$ cannot be solved in time $(2-\epsilon)^k$
    (omitting polynomial factors) unless SETH fails.
  \end{THE}
  \begin{proof}
    We show that an $(2-\epsilon)^k$ algorithm for
    $\WB(\THREECNF,\CHAINS)$ implies an $(2-\epsilon)^n$ algorithm for
    SAT contradicting our assumption.  Let $F$ be a CNF formula with
    $n$ variables. We will transform $F$ into a \THREECNF{} formula
    $F_3$ such that $F$ is satisfiable if and only if $F_3$ has a weak
    $\CHAINS$\hy backdoor set of size at most~$n$.  We obtain $F_3$
    from $F$ using a commonly known transformation that transforms an
    arbitrary CNF formula into a \THREECNF{} formula that is
    satisfiability equivalent with the original formula.  In
    particular, we obtain the formula $F_3$ from $F$ by replacing
    every clause $C=\{x_1,\dotsc,x_l\}$ where $l > 3$ with the clauses
    $\{x_1,x_2,y_1\}$, $\{{\bar y_1},x_3,y_2\}$, $\dotsc$, $\{{\bar
      y_{l-3}},x_{l}\}$, where $y_1,\dotsc,y_{l-3}$ are new variables.
    This completes the construction of~$F_3$. Now, if $F$ is satisfiable and $\tau$ is a satisfying
  assignment of $F$, then the variables of $F$
  form a weak $\CHAINS$\hy backdoor set of size $n$ of $F_3$ with
  witness $\tau$. The reverse direction is immediate since $F_3$ is
  satisfiable, by virtue of having a weak backdoor set, and $F$ is
  satisfiable if $F_3$ is satisfiable. 
    \longversion{

  It remains to show that $F$ is satisfiable if and only if $F_3$ has
  a weak $\CHAINS$\hy backdoor set of size at most $n$. 
  Suppose that $F$ is satisfiable and let $\tau$ be a satisfying
  assignment of $F$. We claim that the variables of $F$
  form a weak $\CHAINS$\hy backdoor set of size $n$ of $F_3$ with
  witness $\tau$. It is easily verified that $F_3[\tau] \in \CHAINS$ and furthermore
  because $F$ is satisfiable so is $F_3$.
  
  The reverse direction follows immediately from the fact that $F_3$
  is satisfiable (because it has a weak $\CHAINS$\hy backdoor set) and
  the fact that the formulas $F$ and $F_3$ are satisfiability
  equivalent by construction.
}
\end{proof}
As the classes $\HORN, \KROM$, and $\ACYC$ contain $\CHAINS$, we have
the following result.
\begin{COR}\sloppypar
  Let $\BBB\in \{\HORN,\KROM,\ACYC \}$. The problem
  $\WB(\THREECNF,\BBB)$ cannot be solved in time $(2-\epsilon)^k$
  (omitting polynomial factors) unless SETH fails.
\end{COR}
 Interestingly, in the case of $\KROM$ the above result holds even if a
 hitting set for all clauses containing $3$ literals is given with the
 input. 
 \shortversion{
The next lower bound is based on the observation that the
\textsc{Vertex Cover} problem can be considered as a special case of
$\WB(\KROM,\ZEROV)$, and on a corresponding lower bound for
\textsc{Vertex Cover} \cite{LokshtanovMarxSuarabh11}.
   \begin{THE}[$\star$]
     $\WB(\KROM,\ZEROV)$, and hence also $\WB(\THREECNF,\ZEROV)$,
     cannot be solved in time $2^{o(k)}$ (omitting polynomial factors)
     unless ETH fails.
   \end{THE}
 }
\longversion{
  \begin{THE}
    $\WB(\KROM,\ZEROV)$, and hence also $\WB(\THREECNF,\ZEROV)$,
    cannot be solved in time $2^{o(k)}$ (omitting polynomial factors)
    unless ETH fails.
  \end{THE}
  \begin{proof}
    To show the theorem we need to recall the following problem.
    \begin{quote}
      \textsc{Vertex Cover}
    
      \emph{Instance:} A graph $G$ and a non-negative integer $k$.
      
      \emph{Parameter:} The integer $k$.  

      \emph{Question:} Does $G$ have a vertex cover of size at most $k$,
      i.e., is there a subset $C$ of the vertices of $G$ of cardinality
      at most $k$ such that $C$ contains at least one endpoint of every
      edge of $G$?
    \end{quote}
    Given a graph $G$ and an integer $k$ we construct in
    linear-time  a \KROM{}
    formula $F$ with $|V(G)|$ variables such that $G$ has a vertex
    cover of size at most $k$
    if and only if $F$ has a weak $\ZEROV$\hy backdoor set of size at most
    $k$. 
    Because \textsc{Vertex Cover} cannot be
    solved in time $2^{o(k)}$ (omitting polynomial factors) unless ETH
    fails~\cite{LokshtanovMarxSuarabh11}[Theorem 3.3] this shows the
    theorem.
    The variables of $F$ are the vertices of $G$ and $F$ consists
    of $1$ clause $\{u,v\}$ for every edge $\{u,v\} \in E(G)$. This
    completes the construction of $F$. It remains to show that $G$ has
    a vertex cover of size at most $k$ if and only if $F$ has a weak
    $\ZEROV$\hy backdoor set of size at most $k$.

    Suppose that $G$ has a vertex cover $C$ of size at most $k$. We
    claim that $C$ is a weak $\ZEROV$\hy backdoor set of $F$ of size at
    most $k$. Let $\tau$ be the assignment of the variables in $C$
    that sets all variables in $C$ to $1$. Then $F[\tau]=\emptyset \in
    \ZEROV$ because $C$ is a vertex cover of $G$.
    
    For the reverse direction suppose that $B$ is a weak $\ZEROV$\hy
    backdoor set of $F$ of size at most $k$ and $\tau$ is a witnessing
    assignment of the variables in $B$. We claim that $B$ is a vertex
    cover of $G$ of size at most $k$. Suppose not, then there is an
    edge $\{u,v\} \in E(G)$ such that $B \cap \{u,v\}
    =\emptyset$. Consequently, $F[\tau]$ contains the clause $\{u,v\}$
    which is not $0$\hy valid 
    contradicting our assumption that $F[\tau] \in \ZEROV$.
  \end{proof}
}
 

Let $\BBB$ be a base class. We say that a polynomial-time algorithm 
${\mathcal A}$ is a \emph{canonical HS reduction for $\BBB$} if
${\mathcal A}$ takes as input an instance $(\SSS,k)$ of \textsc{Hitting Set} over $n$
elements and $m$ sets and outputs an
instance $(F,k)$ of $\WB(\THREECNF,\BBB)$ such that: (a) $F$ has at most
$O(nm)$ variables, and  (b) $\SSS$ has a
hitting set of size at most $k$ if and only if $F$ has a weak
$\BBB$\hy backdoor set of size at most $k$ .

\begin{LEM}\label{lem:canon}
  Let $\BBB$ be a base class. If there is a canonical HS reduction for
  $\BBB$, then the following holds:
  \begin{enumerate}
  \item $\WB(\THREECNF,\BBB)$ is $\W[2]$\hy hard, and
  \item there is no algorithm that solves  $\WB(\THREECNF,\BBB)$ in
    time $O(n^{\frac{k}{2}-\epsilon})$
    unless SETH fails.
  \end{enumerate}
\end{LEM}
\begin{proof}
  Because \textsc{Hitting Set} is $\W[2]$\hy complete and a canonical
  HS reduction is also an fpt-reduction, the first statement of the
  theorem follows. To see the second statement, we first note that it
  is is shown in~\cite[Theorem 5.8]{LokshtanovMarxSuarabh11} that the
  \textsc{Dominating Set} problem cannot be solved in time
  $O(n^{k-\epsilon})$ for any $\epsilon>0$ unless SETH fails (here $n$
  is the number of vertices of the input graph and $k$ is the
  parameter). Using the standard reduction from \textsc{Dominating
    Set} to \textsc{Hitting Set} it follows that \textsc{Hitting Set}
  restricted to instances where the number of sets is at most the
  number of elements, cannot be solved in time $O(n^{k-\epsilon})$ for
  any $\epsilon>0$, where $n$ is the number of elements of the hitting
  set instance and~$k$ is the parameter.  Now suppose that for some
  base class $\BBB$ it holds that $\WB(\THREECNF,\BBB)$ can be solved
  in time $n^{\frac{k}{2}-\epsilon}$ and $\BBB$ has a canonical HS
  reduction. Let $(\SSS,k)$ be an instance of \textsc{Hitting Set}
  with $n_h$ elements and $m_h$ sets. As stated above we can assume
  that $m_h \leq n_h$.  We use the canonical HS reduction to obtain an
  instance $(F,k)$ of $\WB(\THREECNF,\BBB)$ where $F$ has at most
  $O(n_hm_h)\in O(n_h^2)$ variables. We now use the algorithm for
  $\WB(\THREECNF,\BBB)$ to solve \textsc{Hitting Set} in time
  $O((n_h^2)^{\frac{k}{2}-\epsilon})\leq
  O(n_h^{k-\frac{\epsilon}{2}})$ which contradicts our assumption that
  there is no such algorithm for \textsc{Hitting Set}.
 \end{proof}
 \shortversion{\begin{LEM}[$\star$]\label{lem:red-match}
  There is a canonical HS reduction for \MATCH{}.
\end{LEM}}
\longversion{\begin{LEM}\label{lem:red-match}
  There is a canonical HS reduction for \MATCH{}.
\end{LEM}}
\begin{proof}
  Let $(\SSS, k)$ be an instance of \textsc{Hitting Set}
  with $\SSS=\{S_1,\dots,S_m\}$ and $V=\bigcup_{i=1}^m
  S_i=\{x_1,\dots,x_n\}$. We write $S_i=\{x_i^1,\dots,x_i^{q_i}\}$,
  where $q_i=\Card{S_i}$. We construct
  in linear-time
  a \THREECNF{} formula $F$ with $|V|+\sum_{1 \leq i \leq m}(q_i-1)
  \leq n+nm \in O(nm)$ variables
  such that $\SSS$ has a hitting set of size
  at most $k$ if and only if $F$ has a weak $\MATCH$\hy backdoor set of size
  at most $k$. 
  
  The variables of $F$
  consist of the elements of $V$ and additional variables $y_i^j$ for 
  every $1 \leq i \leq m$ and $1 \leq j < q_i$.
  We let $F=\bigcup_{i=0}^m F_i$ where the formulas $F_i$ are defined as
  follows. $F_0$ consists of $n$ binary clauses $\{\bar{x}_1,x_2\}$,
  $\{\bar{x}_2,x_3\},\dots,\{\bar{x}_{n-1},x_n\}$,
  $\{\bar{x}_n,x_1\}$.  For $i>0$, $F_i$ consists of the clauses
  $\{y_i^1,x_i^1\}$, $\{\bar{y}_i^1,y_i^2,x_i^2\}$,
  $\{\bar{y}_i^2,y_i^3,x_i^3\},\dots,$
  $\{\bar{y}_i^{q_i-2},y_i^{q_i-1},x_i^{q_i-1}\}$,
  $\{\bar{y}_i^{q_i-1},x_i^{q_i}\}$.  This completes the
  construction of $F$.  

  \shortversion{We can show that} 
  \longversion{We claim that} $\SSS$ has a
  hitting set of size at most $k$ if and only if $F$ has a weak
  $\MATCH$\hy backdoor set of size at most $k$.
\longversion{
  Let $H\subseteq V$ be a hitting set of $\SSS$. We
  show that $H$ is a weak $\MATCH$\hy backdoor set of $F$.  Let
  $\tau\in 2^H$ be the truth assignment that sets all variables in $H$
  to $1$ and consider $F[\tau]$.  Every clause $C\in
  F_0[\tau]$ contains a unique negative literal for some variable $x
  \in V$; we match $C$ to $x$. Now consider $F_i[\tau]$ for $i>0$.
  We observe that $F_i[\tau]\subseteq F_i$. Since $H$ is a hitting
  set, $\tau$ satisfies at least one clause $C_H$ in $F_i$ such that
  $H \cap C_H=\{x_i^j\}$. Hence, for every clause $C \in F_i[\tau]$ either
  there is an $1 \leq \ell < j$ such that $y_i^{\ell} \in C$ or there is
  an $j \leq \ell < q_i$ such that $\bar{y}_i^{\ell} \in C$. This
  constitutes a matching for the clauses in $F_i[\tau]$.

  For the reverse direction, let $B$ be a weak $\MATCH$\hy backdoor set
  of $F$ and let $\tau$ be an assignment of the variables in $B$ witnessing
  this. We claim that the set $B$
  is a hitting set of $\SSS$. Assume to the contrary that for
  some $1\leq i \leq m$ we have $S_i\cap B=\emptyset$. 
  Because $F_0[\tau]$ has at least $|V \setminus B|$
  clauses and each of these clauses contains only variables in $V
  \setminus B$ it follows that every clause in $F_0[\tau]$ must be 
  matched to a variable in $V \setminus B$. Consequently, for every $1
  \leq i \leq m$ the clauses in $F_i[\tau]$ need to be matched to the
  variables $y_i^1,\dots, y_i^{q_i-1}$. Because $S_i \cap B =\emptyset$
  it follows that the number of variables in $\{y_i^1,\dots,
  y_i^{q_i-1}\} \setminus B$ is $1$ less than the number of clauses of
  $F_i[\tau]$. Consequently, $F_i[\tau] \notin \MATCH$ contradicting our
  assumption that $B$ is a weak $\MATCH$\hy backdoor set of $F$. }
  Hence the lemma follows.
\end{proof}

\medskip\noindent
We observe that the known $\W[2]$\hy hardness proofs for
$\WB(\THREECNF,\RHORN)$ and
$\WB(\THREECNF,\QHORN)$~\cite{GaspersSzeider12,GaspersEtal13} are in
fact canonical HS reductions. Hence, together with
Lemmas \ref{lem:canon} and \ref{lem:red-match} we  
arrive at the following result.

\begin{THE}\sloppypar
  Let $\BBB\in \{\MATCH,\RHORN,\QHORN\}$. Then $\WB(\THREECNF,\BBB)$
  is $\W[2]$\hy hard and cannot be solved in time
  $O(n^{\frac{k}{2}-\epsilon})$ for any $\epsilon>0$ unless SETH
  fails.
\end{THE}




\shortversion{\enlargethispage*{5mm}}
\section{Conclusion}
We have initiated a systematic study of determining the complexity of
finding weak backdoor sets of small size of 3CNF formulas for various base
classes. We have given improved algorithms for some of the base
classes through the bounded search techniques.

Our lower bounds are among the very few  known bounds based on the
(Strong) Exponential-Time Hypotheses for parameterized problems where
the parameter is the  solution size (as opposed to some measure of
structure in the input like treewidth).

Closing the gaps between upper and lower bounds of the problems we
considered in this paper, and studying $\WB(\AAA,\BBB)$ for classes
$\AAA$ other than $\THREECNF$ are interesting directions for further research.


\end{document}